\documentclass[10pt,conference,twocolumn]{IEEEtran}

\usepackage{acronym}
\usepackage{xspace}
\usepackage[normalem]{ulem}
\usepackage{soul}

\usepackage{amsmath,amssymb,amscd,amsthm}
\usepackage{array}
\usepackage[utf8]{inputenc}

\usepackage{graphicx}
\usepackage{float}

\usepackage{rotating}
\usepackage{multirow}
\usepackage{tabularx}

\usepackage{tocbasic}
\usepackage{ifthen}
\usepackage{verbatim}
\usepackage{url}
\usepackage[author={we}]{pdfcomment}

\newlength\figureheight
\newlength\figurewidth

\newcommand{\setsmallfigwidth}{\renewcommand{\figurewidth}{0.7\columnwidth}}

\newcommand{\setsmallfigheight}{\renewcommand{\figureheight}{0.6\columnwidth}}
\newcommand{\settinyfigheight}{\renewcommand{\figureheight}{0.4\columnwidth}}

\def\E{\mathbb{E}}
\def\H{\mathbf{H}}

\newacro{BIAWGN}{binary additive white Gaussian noise}
\newacro{BEC}{binary erasure channel}
\newacro{BMS}{binary-input memoryless symmetric}
\newacro{BP}{Belief Propagation}
\newacro{CN}{check node}
\newacro{DD}{degree distribution}
\newacro{DE}{density evolution}
\newacro{EXIT}{extrinsic information transfer}
\newacro{FPT}{first-passage time}
\newacro{LDPC}{low-density parity-check}
\newacro{MAP}{maximum-a-posteriori}
\newacro{MET}{multi-edge type}
\newacro{OU}{Ornstein-Uhlenbeck}
\newacro{RV}{random variable}
\newacro{PD}{peeling decoding}
\newacro{PPD}{parallel peeling decoding}
\newacro{SC-LDPC}{spatially-coupled low-density parity-check}
\newacro{SPD}{sequential peeling decoding}
\newacro{VN}{variable node}

\newcommand{\BIAWGN}{\ac{BIAWGN}\xspace}
\newcommand{\BEC}{\ac{BEC}\xspace}
\newcommand{\BMS}{\ac{BMS}\xspace}
\newcommand{\BP}{\ac{BP}\xspace}
\newcommand{\cdf}{cdf\xspace}
\newcommand{\CN}{\ac{CN}\xspace}
\newcommand{\CNs}{\acp{CN}\xspace}
\newcommand{\DD}{\ac{DD}\xspace}
\newcommand{\DE}{\ac{DE}\xspace}

\newcommand{\degone}{deg-1\ }
\newcommand{\FPT}{\ac{FPT}\xspace}
\newcommand{\MAP}{\ac{MAP}\xspace}
\newcommand{\LDPC}{\ac{LDPC}\xspace}

\newcommand{\OU}{\ac{OU}\xspace}
\newcommand{\pdf}{pdf\xspace}

\newcommand{\RVs}{\acp{RV}\xspace}
\newcommand{\PD}{\ac{PD}\xspace}
\newcommand{\PPD}{\ac{PPD}\xspace}
\newcommand{\SCLDPC}{\ac{SC-LDPC}\xspace\acused{LDPC}}
\newcommand{\SPD}{\ac{SPD}\xspace}
\newcommand{\VN}{\ac{VN}\xspace}
\newcommand{\VNs}{\acp{VN}\xspace}

\newcommand{\figref}[1]{Fig.~\ref{#1}}

\renewcommand{\v}[1]{\ensuremath{\mathbf{#1}}}

\newcommand{\M}[1]{\ensuremath{\mathbf{#1}}}
\newcommand{\Set}[1]{\ensuremath{\mathcal{#1}}}

\renewcommand{\d}{\ensuremath{\mathrm{d}}}

\newcommand{\e}{\ensuremath{\operatorname{e}}}

\newcommand{\N}{\ensuremath{\mathbb{N}}}

\newcommand{\R}{\ensuremath{\mathbb{R}}}

\newcommand{\pe}{\ensuremath{\epsilon}\xspace}		

\newcommand{\rate}{\ensuremath{\mathtt{r}}}
\newcommand{\n}{\ensuremath{\mathtt{n}}}	
\newcommand{\code}[1]{\ensuremath{\mathcal{#1}}}
\renewcommand{\L}[1]{\ensuremath{L_{#1}}}	
\renewcommand{\R}[1]{\ensuremath{R_{#1}}}	
\renewcommand{\l}{l}				
\renewcommand{\r}{r}				
\newcommand{\F}{\ensuremath{\mathcal{F}}}	

\newcommand{\lrL}{\ensuremath{{(l,r,L)}}\xspace}

\newcommand{\lr}{\ensuremath{(l,r)}\xspace}
\newcommand{\lrl}{\ensuremath{(l,r,L)}\xspace}
\newcommand{\3}[1]{\ensuremath{{(3,6,#1)}}\xspace}
\newcommand{\4}[1]{\ensuremath{{(4,8,#1)}}\xspace}

\newcommand{\GP}[2]{\ensuremath{\mathcal{N}\left(#1,#2\right)}}

\newcommand{\Prob}[1]{\ensuremath{P\left(#1\right)}}

\newcommand{\probsub}[2]{\ensuremath{p_{#1}\left(#2\right)}}

\newcommand{\Ex}[1]{\ensuremath{\operatorname{\mathbb{E}}\left[#1\right]}}

\newcommand{\Var}[1]{\ensuremath{\operatorname{Var}\left[#1\right]}}
\newcommand{\Cov}[1]{\ensuremath{\operatorname{Cov}\left[#1\right]}}

\newcommand{\msg}[2]{\ensuremath{\mu_{{#1}\rightarrow {#2}}}}

\newcommand{\vpdir}{\ensuremath{ \v{p}_{\text{dir}} }}
\newcommand{\Pdir}[1]{\ensuremath{ P^{\text{dir}}_{{#1}}(\ell) }}
\newcommand{\vpindir}[1]{\ensuremath{ \v{p}_{#1,\text{indir}} }}
\newcommand{\Pindir}[1]{\ensuremath{ P_{#1,\text{indir}}(\ell) }}

\newcommand{\Eminus}[1]{\ensuremath{ E^{-}_{#1}(\ell) }}
\newcommand{\Eplus}[1]{\ensuremath{ E_{#1}^+(\ell) }}

\newcommand{\Neighborsset}[1]{\ensuremath{ \mathcal{N}_{#1} }}

\renewcommand{\S}{\ensuremath{\mathcal{S}}\xspace}	

\newtheoremstyle{def}
  {2ex}
  {2ex}
  {\normalfont\itshape}
  {}
  {\normalfont\bfseries}
  {\newline}
  { }
  {\thmname{#1}\thmnumber{ #2} (\thmnote{#3})}

\newtheoremstyle{the}
  {2ex}
  {2ex}
  {\normalfont\itshape}
  {}
  {\normalfont\bfseries}
  {\newline}
  { }
  {\thmname{#1}\thmnumber{ #2}}

\newtheoremstyle{col}
  {2ex}
  {2ex}
  {\normalfont}
  {}
  {\normalfont\bfseries}
  {\newline}
  { }
  {\thmname{#1}\thmnumber{ #2}}

\theoremstyle{def}
\newtheorem{definition}{Definition}[section]

\theoremstyle{the}
\newtheorem{theorem}{Theorem}[section]

\theoremstyle{the}

 \renewcommand\appendix{\par 
   \setcounter{chapter}{0}
   \setcounter{section}{0}%
   \setcounter{subsection}{0}%
   \setcounter{figure}{0}%
   \renewcommand\thechapter{\Alph{chapter}}%
   \renewcommand\thefigure{\Alph{section}.\arabic{figure}}
 }

\renewcommand{\baselinestretch}{0.96}
\setsmallfigwidth
\settinyfigheight

\begin{document}

\title{Finite-length scaling based on belief propagation\\ for spatially coupled LDPC codes \vspace{-2mm}}
\author{
\IEEEauthorblockN{Markus Stinner}
\IEEEauthorblockA{
	Technical University of Munich, Germany\\
	\texttt{markus.stinner@tum.de}
	}
\and
\IEEEauthorblockN{Luca Barletta}
\IEEEauthorblockA{
	Politecnico di Milano, Italy\\
	\texttt{luca.barletta@polimi.it}
	}
\and
	\IEEEauthorblockN{Pablo M. Olmos}
\IEEEauthorblockA{
	Universidad Carlos III de Madrid, Spain\\
	\texttt{olmos@tsc.uc3m.es}
	}
}
\maketitle

\IEEEpeerreviewmaketitle

\begin{abstract}
The equivalence of \PD and \BP for \LDPC codes over the \acl{BEC} is analyzed.
Modifying the scheduling for \PD, it is shown that exactly the same \VNs are resolved in every iteration than with \BP.
The decrease of erased \VNs during the decoding process is analyzed instead of resolvable equations. 
This quantity can also be derived with density evolution, resulting in a drastic decrease in complexity.
Finally, a scaling law using this quantity is established for spatially coupled \LDPC codes. 
\end{abstract}
\begin{keywords}
finite-length performance, spatially-coupled LDPC codes
\end{keywords}
\section{Introduction}

\acresetall
Recently, it was shown that \SCLDPC codes can achieve the channel capacity of \BMS channels under \BP decoding \cite{KudekarBMSIT, Kumar14}.\acused{LDPC}
The Tanner graph of a block code with $M$ \VNs, referred to as the uncoupled \LDPC code graph, is duplicated $L$ times to produce a sequence of identical graphs, where $L$ is the chain length of the \SCLDPC code.
The different copies are connected to form a chain by redirecting (spreading) certain edges.
The asymptotic analysis of  \SCLDPC code ensembles shows that they exhibit a \BP threshold close to the \MAP  threshold of the uncoupled LDPC code ensemble for sufficiently large $L$ \cite{Lentmaier2010a}.
In addition, \SCLDPC code ensembles can be designed with a linear growth of the minimum distance with $M$ \cite{Mitchell11-2}.
Indeed, the minimum distance growth rate for the coupled \LDPC ensemble is often better than for the uncoupled ensemble \cite{Mitchell2014}. Several families of \SCLDPC code ensembles are compared in \cite{Mitchell2014,ks11} using asymptotic arguments, namely \BP threshold and minimum distance growth rate.

{\let\thefootnote\relax\footnotetext{
Markus Stinner and Luca Barletta were supported by an Alexander von Humboldt Professorship endowed by the
German Federal Ministry of Education and Research.
Pablo M. Olmos  was supported  by Spanish government MEC TEC2012-38800-C03-01 and by Comunidad de Madrid (project 'CASI-CAM-CM', id. S2013/ICE-2845).
He works also with the Gregorio Mara\~n\'on Health Research Institute.
}} 

The performance of finite-length \LDPC codes over a \BEC using a \BP decoder is analyzed in \cite{Luby01, Amraoui2009} by studying an alternative decoder, namely \PD.
Following this approach, the finite-length performance of \SCLDPC code ensembles over the \BEC has recently been analyzed in \cite{OlmosPabloM2014,Stinner2015a}.
To the authors knowledge, the single attempt to generalize the \PD-based finite-length analysis to a general message passing \BP decoder is due to Ezri, Montanari and Urbanke in \cite{Ezri2007,Ezri08}, where scaling laws are conjectured for \lr-regular \LDPC ensembles for the \BIAWGN channel and scaling parameters are derived from the correlation of messages sent within the decoder.
However, applying this analysis to \SCLDPC codes is prohibitively complex.

In this work, we present a finite-length analysis approach based on \BP similar to the one base on \PD in \cite{Amraoui2009}. 
After showing the equivalence of \BP and a particular form of \PD with modified scheduling, we replace the analyzed random process used to predict the probability of a decoding failure and examine the decrease of erased \VNs per iteration during the decoding process.
We show that our approach predicts the waterfall performance correctly.
Our substitute can also be derived from \BP \DE, which has significantly lower complexity than graph evolution for \PD.

After introducing \PPD, its equivalence to \BP is shown and we derive a graph evolution for \PPD applied to protograph-based \SCLDPC codes.
The properties shown in \cite{Stinner2015a} are discussed using \PPD.
We then introduce and analyze the decrease of unresolved \VNs per iteration and compare it with the number of resolvable \CNs available in each iteration.
Finally, we establish a scaling law based on the decrease of unresolved \VNs and verify the results with simulations.

This paper is structured as follows.
Section \ref{sec:construction} introduces construction and notation.
In Section \ref{sec:decoders}, we introduce a modified form of \PD called \PPD and \BP and discuss their equivalence.
After deriving the graph evolution for \PPD, we show exemplary that properties and differences between code ensembles observed with standard \PD can still be observed from \PPD. 
The analysis of the decoding trajectory based on resolved \VNs per iteration for \PPD and \BP is discussed in Section \ref{sec:deg1vsmean}.
In Section \ref{sec:SM}, we establish a scaling law based on this random process.



\section{\SCLDPC Code Constructions}\label{sec:construction}
We introduce two types of constructions, randomly constructed regular codes as proposed in \cite{KudekarBMSIT,OlmosPabloM2014}, and a construction based on protographs \cite{Stinner2015a}.
After defining the \BEC, we define the residual graph degree distribution after transmission.
\setcounter{paragraph}{0}

We denote vectors and matrices with $\v{v}=(v_1,v_2,\dots,v_n)$ and \M{M}, respectively.
Extending the notation for unit vectors, $\v{e}_{ i,j,k}$ is a vector where all entries are zero except the entries at positions $i$,$j$, and $k$ which are $1$, whereas \v{1} denotes a vector where all entries are ones.
We also define $\v{a}^\v{v}=a_1^{v_1}a_2^{v_2}\dots$, where the exponent is set to $0$ for any $v_i\leq 0$. 
$|\v{v}|=\sum_{i=1}^n v_i$ as $L^1$ norm is the sum of all entries of $\v{v}$.
Denote by $X=(X_1,\dots,X_n)$ a vector of $n$ \RVs. 

The message passed from \CN $c$ to \VN $v$ in iteration $\ell$ is denoted with $\msg{c}{v}(\ell)$.
We denote the set of \VNs (\CNs) connected to a specific \CN $c$ (\VN $v$) with \Neighborsset{c} (\Neighborsset{v}).

\paragraph{Randomly Constructed \texorpdfstring{$\lrL_u$}{lrLu} \SCLDPC Codes}
Let there be $L$ uncoupled \lr regular \LDPC codes where $l$ is the \VN degree and $r$ is the \CN degree, $r\leq l$, $\tfrac{l}{r}\in\N$.
Each of the $L$ \lr regular \LDPC codes has $M$ \VNs and $\tfrac{l}{r}M$ \CNs and the codes are placed at $L$ consecutive positions.
The $\lrL_u$ code is obtained by spreading $l$ edges of each \VN along consecutive positions, so that each \VN at position $u$ is connected to a \CN at positions $u, \dots, u+l-1$ and a chain of connected codes is obtained as depicted with an example in \figref{fig:sc-36-random}.
When the \CNs at each position are chosen at random, their maximum degree is fixed to $r$.
However, there is randomness in the number of connections to \VNs of a specific position.
Note that there are $(l-1)$ additional positions at the end of the chain of coupled codes without any \VNs but with \CNs connected to \VNs of codes on previous positions of the chain.
For large $L$, the code rate tends to $\rate_{\lrL_u}=1-\tfrac{l}{r}$. 

\begin{figure}[htb]
\footnotesize
\centering
\includegraphics{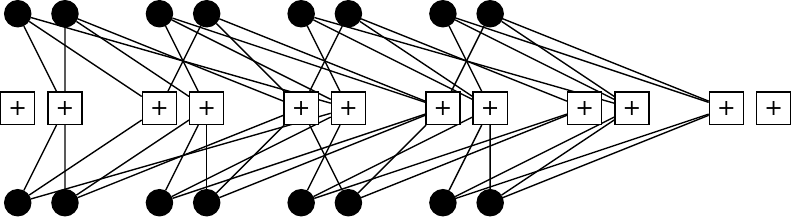}
\caption{Construction of the $\lrL_u=(3,6,4)_u$ coupled protograph with $M=4$ \VNs per code.
The \CNs at positions $3$ and $4$ are regular and have degree $r=6$.}
\label{fig:sc-36-random}
\end{figure}

\paragraph{Protograph-based \lrL SC-LDPC Codes}
The protographs as proposed by Thorpe in \cite{Thorpe2005} are first copied $N$ times before edges of the same type are permuted to avoid small cycles in the resulting code.
Such a protograph can be represented compactly by its bi-adjacency matrix $\M{B}$, called the \emph{base matrix}.
Every $1$ in \M{B} is replaced by an $N\times N$ permutation matrix\footnote{Entries $>1$ in $\M{B}$, represent multiple edges between a pair of specific node types. These entries are replaced by a sum of $N\times N$ permutation matrices.}.
With $v$ \VNs in the protograph, $M=Nv$ \VNs are obtained.
All possible matrices \M{H} derived from all possible combinations of $N\times N$ permutation matrices give a code ensemble.
The \emph{design rate} $\rate$ of this code ensemble can be directly computed from the protograph since the Tanner graph of $\mathbf{H}$ inherits the degree distribution (DD) and graph neighborhood structure of the protograph.

\begin{figure}[htb]
\footnotesize
\centering
\includegraphics{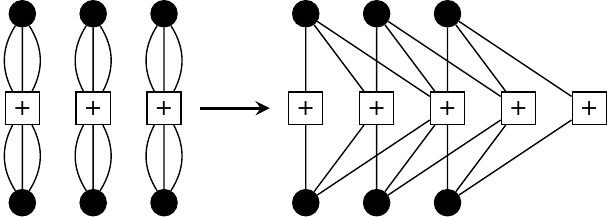}
\caption{Construction of the $\lrL=(3,6,3)$ coupled protograph.}
\label{fig:sc-36}
\end{figure}

We couple \lr-regular LDPC codes according to \cite{Lentmaier2010a} with $\frac{r}{l}=k\in\mathbb{N}$.
$L$ protographs are then connected to an \lrl coupled protograph by connecting each \VN at position $u$, $1\leq u\leq L$ to the \CNs at  positions $u,\dots,u+l-1$ as shown in \figref{fig:sc-36} for $L=3$.
Each uncoupled protograph has $v=k$ \VNs and one \CN so that we obtain $M=2N$ \VNs per coupled code after lifting the construction.
We obtain a code length $\n= kLN$ bits and obtain a code rate of 
$\rate_{\lrl}=1-\tfrac{(L+l-1)}{kL}$.



\subsection{The Binary Erasure Channel}\label{sec:transmission:bec}
Denote by $X(t)\in\{0,1\}$ the binary channel input at a discrete time $t$ and the corresponding channel output $Y(t)\in\{0,1,\Delta\}$ where $\Delta$ denotes an erasure.
We drop the indices for time where possible and use $s$ if $Y\in\{0,1\}$ is known and solved.
A symbol is erased during transmission with probability $\epsilon$ so that we have $\Prob{Y(t)=\Delta}=\epsilon$.
\begin{figure}[htb]
\centering
\includegraphics{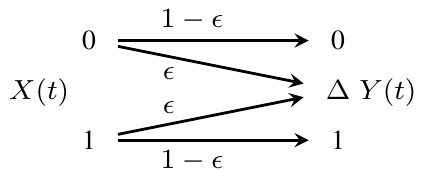}
\caption{The binary erasure channel at time instance $t$.}
\label{fig:bec}
\end{figure}
With uniformly distributed input $X$, i.e. $\Prob{X=1}=\Prob{X=0}=\tfrac{1}{2}$, the capacity of the \BEC is $C_{\text{\BEC}}=H(Y)-H(Y|X)=1-\epsilon$.

\subsection{Degree Distribution of the Residual Graph After Transmission}

After transmitting \VNs $v_1,\dots,v_n$ over a \BEC, a certain fraction of \VNs is erased as illustrated in \figref{fig:24transmission}.
The graph of edges connected to erased \VNs is often called residual graph since it represents the set of parity equations which remain to be solved to recover the whole codeword.

\begin{figure}[htb]
\centering
\includegraphics{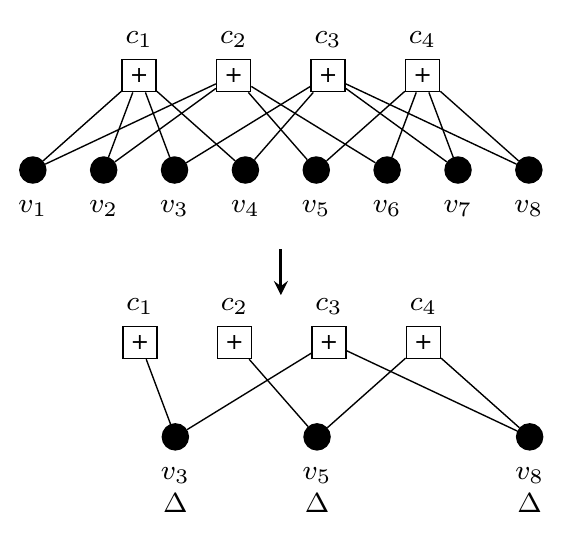}
\caption{Residual graph after transmission over a BEC.}
\label{fig:24transmission}
\end{figure}

To define the \DD, we label each edge in the protograph connecting a different pair of nodes with $1$ to $m$.
In the Tanner graph representation of the parity check matrix \M{H} of a code generated by a protograph, we denote the type a particular edge was copied from with $j\in\{1,2,\ldots,m\}$.
We denote the type of a \VN with $\v{v}=(d_1,\dots,d_m)$, where $d_j\in\N_0$ represents the number of edges of type $j$ connected to this \VN type.
Similarly, we define the type of a \CN by $\v{c}$ and represent the number of \VNs (\CNs) of  type $\v{v}$  ($\v{c}$) in the Tanner graph of $\H$ with $\L{\v{v}}$ ($\R{\v{c}}$).
The set of \VN (\CN) types in the graph is denoted by $\F_v$ ($\F_c$).

\figref{fig:residualtypes} shows the labeling of an uncoupled $(2,4)$  \LDPC ensemble as an example and possible outcomes of \CNs in the residual graph after transmission.
Observe that there are two edges of type $1$ and two of type $2$.  
As discussed in \cite{Stinner2015a}, many combinations of known and unknown edges are possible for a \CN type before transmission. 
We denote the set of \CN types after transmission with $\overline{\F}_c$.
Note that there are no additional \VN types after transmission so that the set of \VN types in the residual graph is still $\F_v$.

\begin{figure}[htb]
\centering
\includegraphics{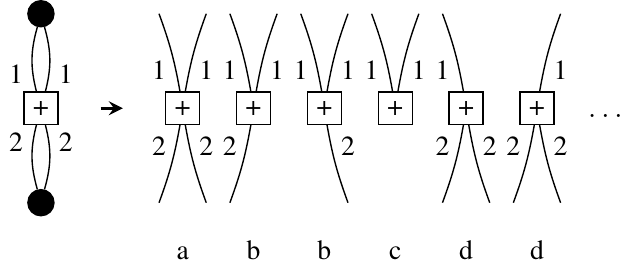}
\caption{CN types after transmission.}
\label{fig:residualtypes}
\end{figure}
\section{Decoders for the Binary Erasure Channel}\label{sec:decoders}
Consider transmission over a \BEC with erasure probability~$\epsilon$.
In this section, \SPD, \PPD and \BP are introduced.
We formulate \SPD and \PPD in terms of messages sent, and split the messages sent from \VNs to \CNs into \emph{forward} messages if the \CN is connected to the residual graph, and \emph{backward} messages if all other \VNs connected to the respective \CN are already resolved.

\subsection{Sequential Peeling Decoder}
All known \VNs of the graph and their connected edges are removed to obtain the residual graph.
A \degone \CN is a \CN of this residual graph with only one connected unknown \VN.
In every iteration $\ell$, a single \degone \CN $c$ is chosen and the connected \VN $v$ is resolved.
We remove $v$ and all adjacent edges to \CNs $c'\in\Neighborsset{v}$ from the residual graph.
To calculate messages, we have 
\begin{align}
  \msg{c}{v}(\ell)=\begin{cases}
              s, \text{ if all } \msg{v'}{c}(\ell)=s, v'\in\Neighborsset{c}\setminus\{v\}\\
              \Delta,\, \text{else,}
             \end{cases}
\end{align}
for any \VN $v$ and any \CN $c$ of the graph.
\begin{figure}[!h]
\centering
\includegraphics{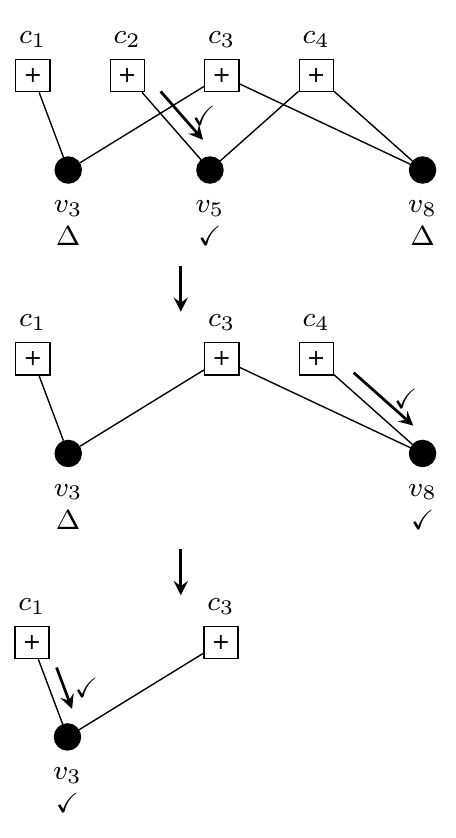}
\caption{Decoding using \SPD.}
\label{fig:24seqPD}
\end{figure}
If any $\msg{c}{v}(\ell),c\in\Neighborsset{v}$ is resolved, then $v$ is resolved and becomes fixed, and all outgoing $\msg{v}{\Neighborsset{v}}(\ell+i)$ stay resolved in further iterations $\ell+i$:
\begin{align}
  \msg{v}{c}(\ell+1)=\begin{cases}
              s, \text{ if }\msg{c'}{v}(\ell)=s\text{ for some } c'\in\Neighborsset{v}\\
              \Delta,\, \text{else.}
             \end{cases}
\end{align}
We call $\msg{v}{c}(\ell+1)=f(\msg{c'}{v}(\ell)),c\in\Neighborsset{v}\setminus\{c'\}$ the messages \emph{passed forwards } from $v$ and $\msg{v}{c}(\ell+1)=f(\msg{c}{v}(\ell))$ the message \emph{fed backwards}.
Note that with \SPD, $\msg{v}{c}(\ell+1)$ is a function of $\msg{c}{v}(\ell)$.
An example of \SPD is illustrated in \figref{fig:24seqPD} for the residual graph shown in \figref{fig:24transmission}.
Since in every iteration the \degone \CN to be resolved is picked randomly, the sequence of residual graphs for several decoding realizations may differ for a given transmission realization.

\subsection{Parallel Peeling Decoder}
\PPD uses a different scheduling than \SPD. 
Instead of resolving only a single \degone \CN per iteration, all available \degone \CNs are resolved.

\begin{figure}[h]
\centering
\includegraphics{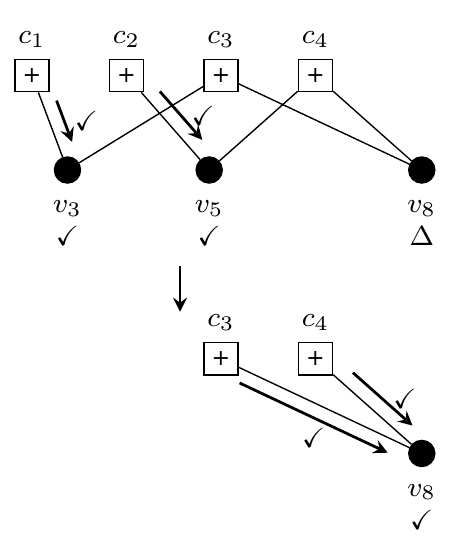}
\caption{Decoding using \PPD.}
\label{fig:24parPD}
\end{figure}

\figref{fig:24parPD} depicts \PPD iterations for the residual graph of \figref{fig:24transmission}.
Note that \PPD is deterministic since for a given transmission realization, all available \degone \CNs are resolved in every step and thus the sequence of residual graphs does not differ.

\subsection{Belief Propagation Decoder}\label{sec:decoders:bp}

For \BP decoding, we apply iterative message passing as described in \cite{Richardson2008,Kschischang2001}.
After initializing the \VNs with their corresponding channel output after transmission, messages are passed from \VNs to their adjacent \CNs and back again as illustrated in the example in \figref{fig:24BP}.
The \CN function is identical to \PPD.
However, every $\msg{v}{c}(\ell)$ from \VN $v$ to any adjacent \CN $c$ depends only on messages received from all other adjacent \CNs in $\Neighborsset{v}\setminus\{c\}$ than $c$:
\begin{align*}
  \msg{v}{c}(\ell+1)=\begin{cases}
              s \text{, if }\msg{c'}{v}(\ell)=s\text{ for some } c'\in\Neighborsset{v}\setminus\{c\}\\
              \Delta, \text{else.}
             \end{cases}
\end{align*}
Since $\msg{v}{c}(\ell+1)$ fed backwards does not depend on $\msg{c}{v}(\ell)$, the messages sent in both directions can differ.
\begin{figure}[t]
\centering
\includegraphics{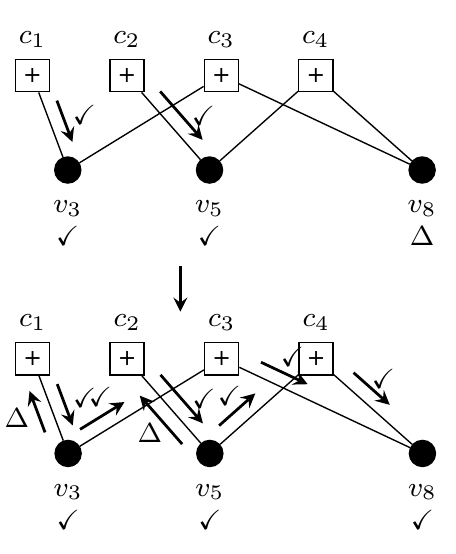}
\caption{Decoding using BP.}
\label{fig:24BP}
\end{figure}
\BP for \figref{fig:24transmission} is depicted in \figref{fig:24BP}. 
Since all other messages are resolved, we only show the messages in the remaining residual graph.
\BP is also deterministic.

\subsection{Equivalence of \SPD and \PPD}\label{sec:decoders:seqpdvsparpd}


We compare \SPD and \PPD.
We define stopping sets and show that with infinitely many iterations, both decoders give the same result.

\begin{definition}[Stopping Set \cite{Di2002}]
A stopping set \S is a subset of the set of all \VNs of the code $\code{C}$ such that all neighbor \CNs of \S are connected to \S at least twice.
\end{definition}

\setcounter{paragraph}{0}
\paragraph{Same result for \SPD and \PPD}
\SPD and \PPD always obtain the same decoding result.
This can be explained intuitively since decoding on the \BEC equals solving a linear system of equations.
Thus, \SPD and \PPD always fail for underdetermined parts of the system of equations,  and therefore obtain the same decoding result as discussed in \cite{Di2002}.

\paragraph{Messages sent within \S stay erased}
Consider using \SPD, \PPD, and \BP. 
The \CN functions of all three decoders are identical and the decoders can resolve a \VN $v$ connected to a \CN $c$ with $\msg{c}{v}(\ell)$ only if all other incoming messages to $c$ from \VNs $v'\in\Neighborsset{c}\setminus\{v\}$ are known in iteration $\ell$. 
All \CNs adjacent to \VNs of \S are connected to unresolved \VNs of \S at least twice.
Thus, none of the decoders is able to resolve any message $\msg{c}{v}(\ell),v\in\S,c\in\Neighborsset{v}$ sent to any of the adjacent \VNs $v\in\Neighborsset{c}$ in any iteration $\ell$.
Note that not only messages passed within a stopping set \S can never be resolved, but also messages sent from \S to the rest of the residual graph can never be resolved during the decoding process.

Note that \PPD consists of a series of (sampled) states of a particular realization of the random \SPD.

\subsection{Equivalence of \PPD and \BP}\label{sec:decoders:equivalence}

According to \cite{Richardson2008}, the erasure probability of every \VN is monotonically decreasing during the decoding process.
Di showed in \cite{Di2002} that an iterative decoder will obtain a specific solution for a given realization of a codeword transmitted over the \BEC. 
However, the behavior of different decoders during decoding is not discussed in detail.


\begin{theorem}[Exactly Equivalent Recovery]
Given any transmission realization over the \BEC, \PPD and \BP recover exactly the same \VNs at each iteration.
\end{theorem}

The full proof is given in Appendix \ref{sec:appendix:parpdvsbp}.
As an outline, we first show that without stopping sets, both \PPD and \BP recover exactly the same erased \VNs using messages only passed forwards.
We then show that the message fed backwards using the \PPD does not resolve any additional \VNs, and thus the two decoders resolve the same \VNs in every iteration if starting with the same residual graph.

\subsection{Graph Evolution of \degone \CNs During \PPD}\label{sec:parallel_graphvevolution}
Traditionally, the statistical evolution of \degone \CNs during PD is used to analyze the finite-length behavior of  a given code ensemble \cite{Amraoui2009,OlmosPabloM2014}.
In every iteration, \PPD removes every \degone \CN, the respective adjacent \VN and all attached edges, i.e. the probability of removing any \degone \CN during an iteration is $1$ unlike for \SPD.
The normalized \DD is defined as
\begin{align}\label{eq:parallelPDnorm}
 \l_{\v{v}}(\ell)\doteq\frac{\L{\v{v}}(\ell)}{M},\qquad \r_{\v{c}}(\ell)\doteq\frac{\R{\v{c}}(\ell)}{M},
\end{align}
for all $\v{v}\in\F_v, \v{c}\in\overline{\F}_c$.
We obtain the sum of \degone \CNs by 
\begin{align}
c_1(\ell)=\sum_{j=1}^m \r_{\v{e}_j}(\ell).                                   
\end{align}
As for \SPD, the threshold $\epsilon^*$ of an SC-\LDPC code ensemble is given by the maximum $\epsilon$ for which the expected sum of \degone \CNs $\Ex{c_1(\ell)}$ is positive for any $\ell$ during the decoding process, i.e. $\Ex{c_1(\ell)}>0,\ell\in(0,\Omega]$ where $\Omega$ is the stopping time such that all \VNs are recovered.

The average error probability is also dominated by the probability that $c_1(\ell)$ \emph{survives} as discussed in \cite{Amraoui2009}.
We modify and extend the expected graph evolution of the \SPD to adapt it for \PPD and the analysis must take into account solving multiple \degone \CNs in every iteration as explained in Appendix \ref{sec:appendix:meanstepparallel}.
For each iteration, we apply the following steps:
\begin{itemize}
 \item For each \VN, calculate the probability that it is connected to any \degone \CN;
 \item Erase all connected \degone \CNs;
 \item Update all other connected \CNs.
\end{itemize}

\begin{figure}[t]
\centering
\includegraphics{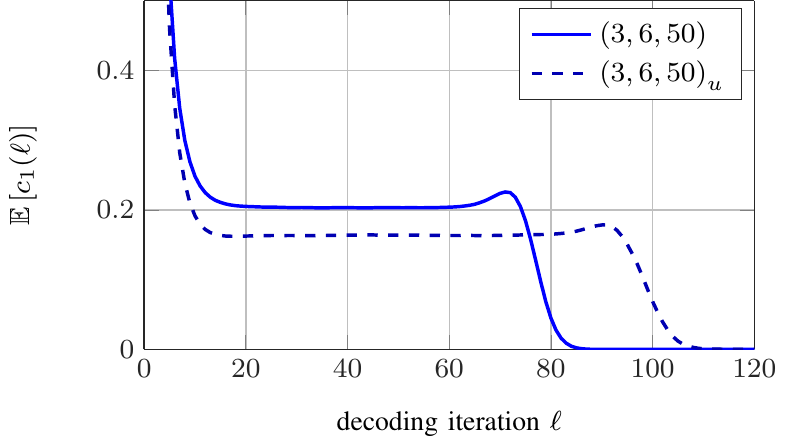}\hspace{2em}
\caption{Simulated $\Ex{c_1(\ell)}$ for \3{50} and $\3{50}_u$ at $\epsilon=0.45$.
}
\label{fig:parallelPD_mean_deg1_randomvsstructured}
\end{figure}

Since we focus only on the graph evolution for protograph-based codes,
\figref{fig:parallelPD_mean_deg1_randomvsstructured} compares simulated results for $\Ex{c_1(\ell)}$ at $\epsilon=0.45$ for the \3{50} and the $\3{50}_u$ ensemble averaged over $5000$ transmissions.
We refer to the phase where $c_1(\ell)$ is constant as \emph{critical} phase, and use for it the symbol $c_1(*)$.
Observe that $c_1(*)$ is lower for the $\3{50}_u$ ensemble, similar to what was observed for \SPD in \cite{Stinner2015a}.
Thus, more iterations are needed to recover all \VNs which results also in a longer critical phase.
Having a longer critical phase with lower $c_1(*)$ also explains why the $\3{50}_u$ ensemble performs worse as observed in \cite{Stinner2015a}.

\section{Towards Finite-Length Analysis for \BP}\label{sec:deg1vsmean} 
The smaller $(\epsilon^*-\epsilon)$ is, the more iterations are needed for decoding.
As in \cite{Ezri08}, we define the normalized time $\tau$ for the \PPD and \BP as 
\begin{align}
\tau=\ell\cdot(\epsilon^*-\epsilon)
\end{align}
so that the previous iteration is $\tau-(\epsilon^*-\epsilon)$.

The initial channel erasure probability of a \VN is denoted with $\epsilon$.
We extend the notation and denote by $\epsilon(\tau)$ the average erasure probability of \VNs in iteration $\tau$.
Note that $\epsilon(\tau)$ can be calculated in two ways.
On one hand, $\epsilon(\tau)$ can be calculated with the evolution described in Section \ref{sec:parallel_graphvevolution}.
On the other hand, we can also apply \DE for \BP which is discussed in Appendix \ref{sec:appendix:DE}.
\DE for the \BEC has low complexity and is a common analysis tool.

\begin{figure}[b]
\centering

\includegraphics{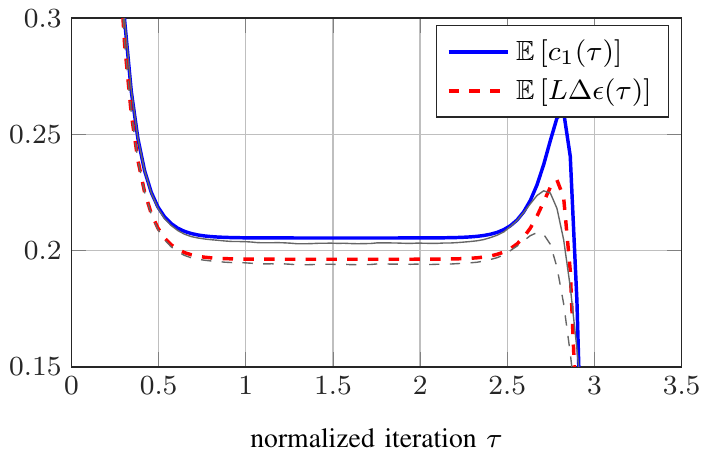}
\caption{ $\Ex{c_1(\tau)}$ and $\Ex{L\Delta\epsilon(\tau)}$ for \3{50}. Simulated results are plotted as gray lines for verification.
}
\label{fig:parallelPD_deg1vsdelta_36}
\end{figure}

In general, the complexity of calculating the \SPD graph evolution is much higher than for \DE.
For each \CN type $\v{c}\in\F_c$, we have $|\overline{\F}_\v{c}|=2^{d_c}$ types after transmission considered during the graph evolution and only $d_c$ erasure probabilities.
In total, there are $|\overline{\F}_c|$ \CN types to track for graph evolution whereas for \DE, we only need to track the erasure probabilities of $m=\sum_{\v{c}\in\F} d_c$ edge types.
As an example, consider again the $\3{3}$ example from \figref{fig:sc-36} with $M=4000$ and $\pe=0.45$.
We have $|\overline{\F}_c|=104$ \CN types in $\overline{\F}_c$ and therefore $104$ \CN types need to be tracked during $\ell=4500$ iterations of the graph evolution.
Using \DE, there are only $18$ erasure probabilities to track for the respective edge types of the code ensemble and we need only $6$ iterations.

Denote with $\Delta\epsilon$ the change of $\pe$ between two consecutive iterations.
In order to compare $c_1(\tau)$ with the decrease of erased \VNs, we propose to study a random process based on $\epsilon(\tau)$ 
\begin{align}
L\Delta\epsilon(\tau)=L[\epsilon(\tau-(\epsilon^*-\epsilon))-\epsilon(\tau)]  \label{eq:delta_epsilon}
\end{align}
which refers to the number of variable nodes resolved per \BP iteration of a \SCLDPC code ensemble with $L$ coupled codes, normalized by $M$ as done in \eqref{eq:parallelPDnorm}.
This decrease in erasure probability is closely related to the convergence \emph{speed} \cite{Aref2013} where the velocity of the propagation wave is measured.
There, in place of the total decrease in erasure probability per iteration, it was analyzed how many iterations $\ell+I,\, I\in\N$, it takes until the erasure probabilities of the \VNs at code position $u+1$ are reduced to the erasure probabilities of the \VNs at positions $u$ in iteration $\ell$.

\figref{fig:parallelPD_deg1vsdelta_36} shows the analytical predictions of $\Ex{c_1(\tau)}$ and $\Ex{L\Delta\epsilon(\tau)}$ verified with $5000$ transmissions for a \3{50} code with $M=4000$.
Simulation and prediction fit accurately.
Observe $\Ex{c_1(\tau)}\geq \Ex{L\Delta\epsilon(\tau)}$ in the critical phase.
This motivates the use of $\Delta\epsilon(\tau)$ as a surrogate for $c_1(\tau)$ and we can state that in fact, 
$c_1(\tau-(\epsilon^*-\epsilon))$ is an upper bound of $L\Delta\epsilon(\tau)$.
\begin{theorem}[Upper Bound of Survival Process]
Assume transmission over a BEC. For any \PPD and \BP process, $c_1(\tau-(\epsilon^*-\epsilon))$ is an upper bound on $L\Delta\epsilon(\tau)$:
\begin{align}
 L\Delta\epsilon(\tau)\leq c_1(\tau-(\epsilon^*-\epsilon)).
\end{align}

\end{theorem}
\begin{proof}
We normalize with respect to $M$ and $\tau$.
$L\Delta\epsilon(\tau)$ \VNs are resolved in an iteration.
\VNs can only be resolved if they are connected to a \degone \CN and every \degone \CN can resolve a single \VN of the residual graph as discussed in Section \ref{sec:decoders:equivalence}.
We know that in iteration $\tau$, $c_1(\tau-(\epsilon^*-\epsilon))$ \degone \CNs will be resolved.
If every of these resolved \CNs is connected to a different \VN, $c_1(\tau-(\epsilon^*-\epsilon))$ \VNs will be resolved.
If $j\geq1$ \degone \CNs are connected to any \VN $v$, $c_1(\tau-(\epsilon^*-\epsilon))-\tfrac{(j-1)}{M}$ \VNs will be resolved.
Thus, it is not possible to resolve more \VNs.
\end{proof}

%
%

\subsection{Mean Evolution During the Decoding Process}


As mentioned before, \DE can be applied in order to compute $\E[\epsilon(\tau)]$ and thus to evaluate the process $\Ex{L\Delta\epsilon(\tau)}$. 
The \DE solution to  $\Ex{L\Delta\epsilon(\tau)}$ can be used not only to characterize the asymptotic behavior, i.e. to compute the decoding threshold of the ensemble, but also to determine  quantities needed to assess the finite-length performance of the code.
Similar to \cite{Amraoui2009}, the average error probability over the ensemble of codes is dominated by the probability that the process $\Delta{\epsilon}(\tau)$ \emph{survives}, i.e. it does not hit the zero plane.
Therefore, characterizing the critical phase and $\Delta{\epsilon}(\tau)$ at that time determines the \SCLDPC finite-length performance.

\subsection{Analysis of the Moments of \texorpdfstring{$\Delta\epsilon(\tau)$}{Delta epsilon}}\label{sec:results_mean}

We approximate the expected fraction of recovered \VNs per coupled code during the critical phase, for $\epsilon$ close to $\epsilon^*$, with 
\begin{align}
 L\Delta\epsilon(\tau)\approx \gamma(\epsilon^*-\epsilon),
\end{align}
where $\gamma$ is a positive constant.
As shown in \figref{fig:delta_eps}, this approximation is reasonable and accurate.

\begin{figure}[t]
\centering
\vspace{0.57em}
\includegraphics{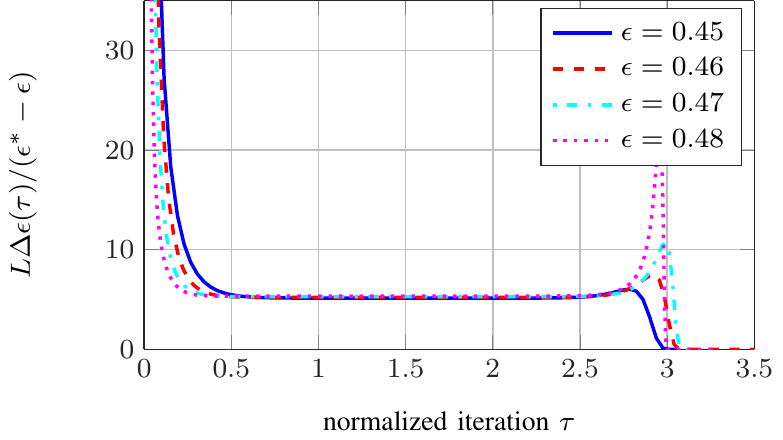}\hspace{2em}
\caption{$\Ex{L\Delta{\epsilon}(\tau)/(\epsilon^*-\epsilon)}$ for the \3{50} ensemble. 
}
\label{fig:delta_eps}
\end{figure}




We approximate $\Var{L\Delta{\epsilon}(\tau)}$ during the decoding process with $\delta=M\Var{L\Delta{\epsilon}(*)}$ during the critical phase.
Observe in \figref{fig:var_covar_delta_eps} (a) that $\delta$ is constant for different $(\epsilon^*-\epsilon)$.



We also examine $\Cov{L\Delta{\epsilon}(\tau),L\Delta{\epsilon}(\tau')}$ during the decoding process.
\figref{fig:var_covar_delta_eps} (b) shows $\Cov{L\Delta\epsilon(\tau),L\Delta\epsilon(\tau')}$ for $\tau'=1.5$ and $\tau'=1.75$ during the critical phase of the decoding.
Observe that the decay of $\Cov{L\Delta{\epsilon}(\tau),L\Delta{\epsilon}(\tau')}$ behaves similar for both values of $\tau'$ during the critical phase.



\begin{figure}[t]
\hspace{-0.6em}\includegraphics{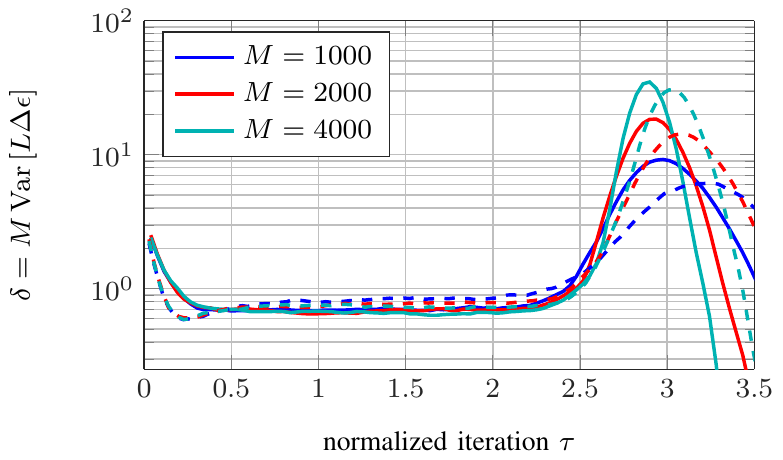}
\\ \centering{(a)} \\
\hspace{-2.45em}\includegraphics{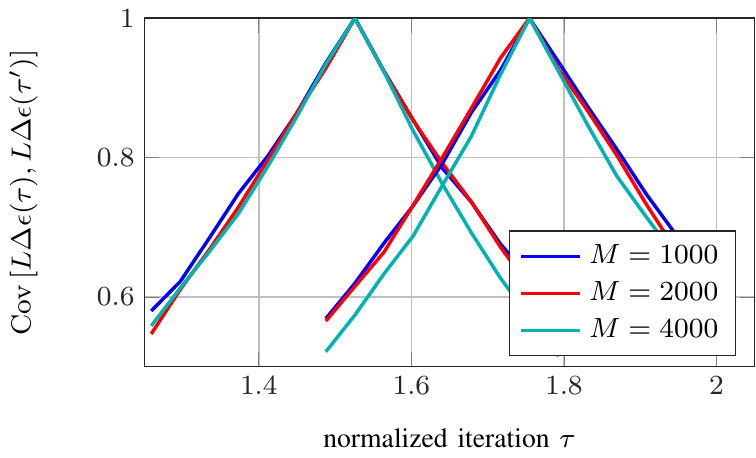}
\\ \centering{(b)}
\caption{In (a), simulated and normalized $M\Var{L\Delta\epsilon(\tau)}$ for the \3{50} ensemble with $M=\{1000,2000,4000\}$ and $\epsilon=0.45$ (solid), respectively $\epsilon=0.46$ (dashed) from $5000$ simulated transmissions.
In (b), simulated $\Cov{L\Delta\epsilon(\tau),L\Delta\epsilon(\tau')}$ for the \3{50} ensemble with $M=\{1000,2000,4000\}$ and $\epsilon=0.45$ for $\tau'=1.5$ and $\tau'=1.75$.
}
\label{fig:var_covar_delta_eps}
\end{figure}

\subsection{Stability of the Process During the Critical Phase}
We now compare the ratios
\begin{align}
 \alpha_{c_1}(\tau)&=\frac{\Ex{c_1(\tau)}}{\sqrt{\Var{c_1(\tau)}}},
 &\alpha_{\Delta\epsilon}(\tau)&=\frac{\Ex{\Delta\epsilon(\tau)}}{\sqrt{\Var{\Delta\epsilon(\tau)}}}.
\end{align}
Note that $\alpha_{L\Delta\epsilon}(\tau)=\alpha_{\Delta\epsilon}(\tau)$ so that the normalization with $L$ is not needed.
The two quantities are plotted for the \3{50} code with $M=4000$ in \figref{fig:compare_alpha}.
Observe that during the whole decoding process, they are very close.
In the critical phase, we have $\alpha_{c_1}=6.049$ and $\alpha_{\Delta\epsilon} = 6.376$ as listed in Table \ref{tab:alphas}.

\begin{figure}[b]
\centering
\includegraphics{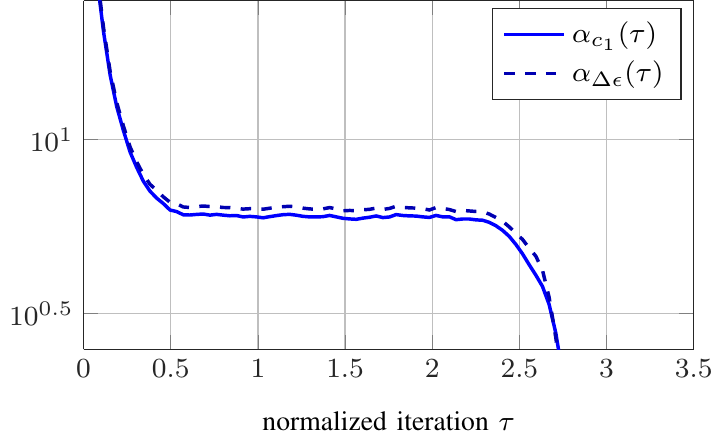}
\caption{ $\alpha_{c_1}(\tau)$ and $\alpha_{\Delta\epsilon}(\tau)$ for \3{50}.
}
\label{fig:compare_alpha}
\end{figure}

\section{Statistical Models for Failure Probabilities}\label{sec:SM}
Given the above results, it is reasonable to follow the argumentation for $c_1(\tau)$ and \SPD in \cite{Amraoui2009} and to model $\Delta\epsilon$ as a Markov process.
For SC-\LDPC codes, mean and variance are constant in the critical phase as for $c_1$ in \cite{OlmosPabloM2014}, and we extend the model of $\Delta\epsilon$  to a Stationary Gaussian Markov process, more specifically to an \OU process.
Since the decoding fails when $\Delta\epsilon(\tau)$ hits $0$ before the decoding finishes, we consider the \FPT across $0$ of such a process.



\subsection{Ornstein-Uhlenbeck Processes }\label{sec:OU} 
Let $t_1<\dots<t_n$.
A stationary Gaussian Markov Process is also called \OU process if $X(t_1),\dots,X(t_n)$ are jointly distributed as a multivariate Gaussian distribution where mean and variance are constant and $\Cov{X(t+T),X(t)}\propto\exp^{-\alpha T}$ with constant $\alpha>0$, and $T>0$.
The \OU process is obtained by adding a state-dependent drift to a Wiener process \cite{Gardiner2009}:
\begin{align}
\d X(t)=-\Theta(\mu - X(t)) \d t+\sigma \d W(t),
\end{align}
where $W(\cdot)$ denotes a standard Wiener process and $X(\cdot)$ is the unknown.
Let the process start in its mean value so that $\mu=X(0)=x_0$.
Using $f(X(t),t)=X(t)\e^{\Theta t}$\cite{I.Karatzas1991}, we have 
\begin{align}
 \!\! X(t)
 = x_0\e^{-\Theta t} + \mu(1-\e^{-\Theta t}) 
 + \e^{-\Theta t} \int_0^t \sigma \e^{\Theta s} \d W(s),
\end{align}
with $\sigma=\sqrt{2b}$ as in \cite{OlmosPabloM2014}. 
Since $x_0$ is a constant, we have
\begin{align}
 X(t)& \sim \GP{x_0}{\frac{b}{\Theta}}
\end{align}
and the expectation of the mean is $\Ex{X(t)}=x_0$.
Assume that $\min(t, u) = t$.
For the covariance, we have 
\begin{align}
 \Cov{X(t),X(u)} 
 &\approx\frac{b}{\Theta} \left( \e^{-\Theta|t-u|}\right)
\end{align}
for sufficiently large $t+u$.
With the observed quantities from Section \ref{sec:results_mean}, we have
\begin{align}
 \Ex{X(t)}&= x_0=\gamma(\epsilon^*-\epsilon),\\
 \Var{X(t)}&= \frac{b}{\Theta}=\frac{\delta}{M}.
\end{align}

\subsection{First-Passage Time Distribution}\label{sec:FPT}
We analyze the average time when $\Delta\epsilon(\tau)$ takes the value $0$ the first time.
Using the symmetry of \OU processes, we change the initial state to $X(0)=0$ and set a fixed boundary $s=x_0$.
We define the \FPT $T_s$ as the time period before the first crossing in $s$:
\begin{align}
T_{s}=\inf_{t\geq 0}\{t:X(t)\geq s\}.
\end{align}
Denote by $\mu_0=\Ex{T_s}$ the mean \FPT from the zero initial state to the boundary $s$.
Evaluating the \pdf of $T_s$ becomes a complex problem when $\tfrac{s}{b/\Theta}\rightarrow 0$ without any existing closed-form expression \cite{OU-4}.
For a reasonably large ratio $\tfrac{s}{b/\Theta}$, it is shown in \cite{OU-3} that the \pdf of the \FPT converges to an exponential distribution
\begin{align}
 \probsub{T_{s}}{t}\sim \frac{1}{\mu_0} \e^{\tfrac{-t}{\mu_0}},\label{eq:tsdist}
\end{align}
and according to \cite{OU-3,OU-1}, $\mu_0$ can be explicitly calculated with
\begin{align}
 \mu_0=&\frac{\sqrt{2\pi}}{\Theta}\int_0^{\tfrac{s}{\sqrt{b/\Theta}}}\Phi(z)\e^{\tfrac{1}{2}z^2}\d z\nonumber\\
 =&\frac{\sqrt{2\pi}}{\Theta}\int_0^{\tfrac{\gamma}{\sqrt{\delta_1}}\sqrt{M}(\epsilon^*-\epsilon)}\Phi(z)\e^{\tfrac{1}{2}z^2}\d z\label{eq:FPT}
\end{align}
where $\Phi(z)$ is the \cdf of the standard Gaussian distribution.
\subsection{The Scaling Law for \texorpdfstring{$\Delta\epsilon(\tau)$}{Delta epsilon}}\label{sec:scalinglaw}
The smaller $(\epsilon^*-\epsilon)$ is, the more iterations are needed for decoding.
For $\epsilon$ very close to $\epsilon^*$, it is reasonable to model the decoding process as a continuous-time process as done in \cite{Amraoui2009} for \SPD. 
Similar to \cite{Amraoui2009}, the average error probability over the ensemble of codes is dominated by the probability that the process $\Delta{\epsilon}(\tau)$ survives, i.e. it does not hit the zero plane.
Therefore, characterizing the critical points and the expected $\Delta{\epsilon}(\tau)$ at that time determines the SC-LDPC finite-length performance.

\begin{figure}[t]
\centering
\setsmallfigheight
\begin{tabular}{c}
\includegraphics{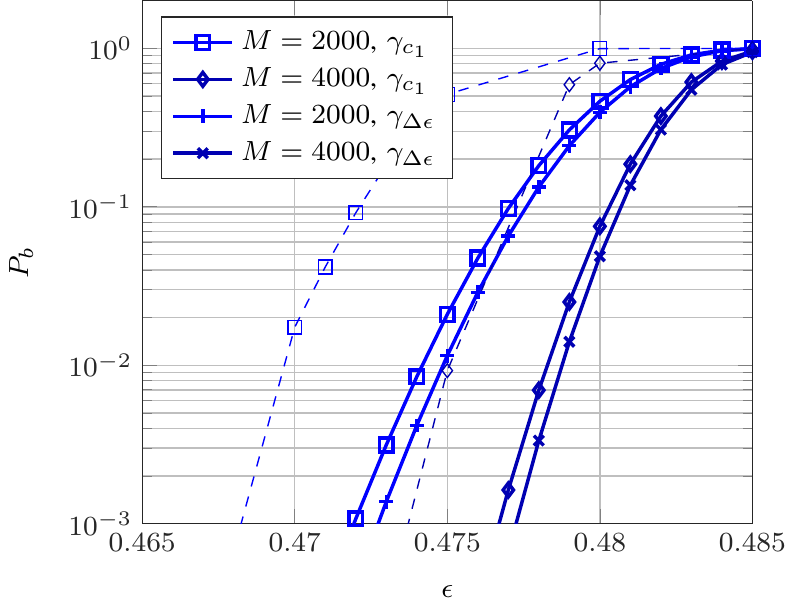}
\\ (a) \\
\includegraphics{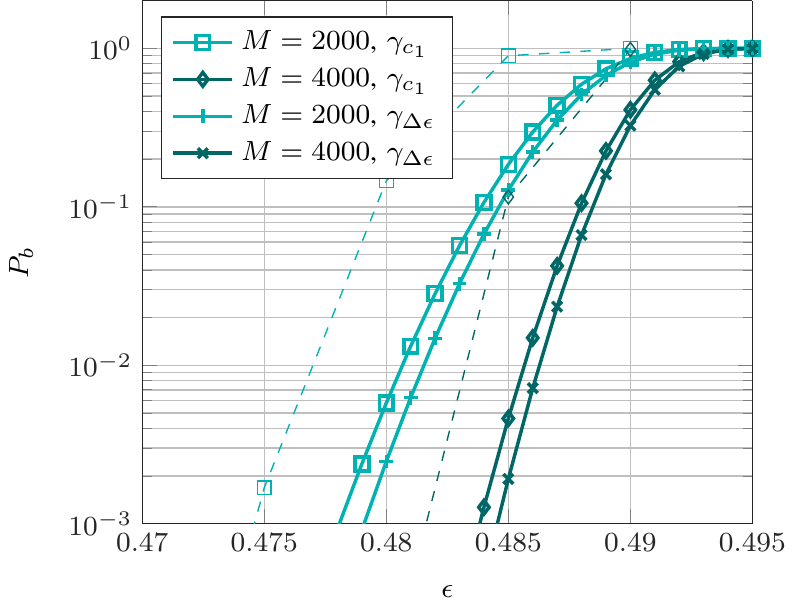}
\\ (b)
\end{tabular}
\caption{
Predicted trajectories for \3{100} in (a) and \4{100} in (b) for $M=\{2000,4000\}$ 
using $\gamma_{c_1}$ and $\gamma_{\Delta\epsilon}$.
Simulated results are included as dashed lines.
}
\label{fig:SL_deg1vsdelta}
\end{figure}

Based on the statistics describing the $\Delta \epsilon(\cdot)$ process, we model its critical phase with an \OU process as described in Section \ref{sec:FPT}.
Denote with $\tau_\text{eff}$ the duration of the critical phase of a decoding process:
\begin{align}
\tau_\text{eff}&=\tau_\text{predicted}-\tau_\text{corr}=t_{\text{predicted}}(\pe^*-\epsilon)-\tau_\text{corr},
\end{align}
where $\tau_\text{predicted}$ is calculated with \DE and $\tau_\text{corr}$ is a correction term taking into account the time before and after the critical phase.
Note that 
$\tau_\text{eff}$ is a function of $L$ and the given ensemble.
The probability of a block error $P_b^*$ equals the probability of a zero-crossing before $\tau_\text{eff}$:
\begin{align}\label{eq:longchain}
P_b^{*}=\Pr(T_s\le \tau_\text{eff})\approx 1-\exp\left(-\frac{\tau_\text{eff}}{\mu_0}\right),
\end{align}
since with \eqref{eq:tsdist}, $T_s$ is approximately exponentially distributed.

\subsection{Comparison to Estimates from \SPD}
The obtained parameters are summarized in Table \ref{tab:alphas}.
\figref{fig:SL_deg1vsdelta} compares the estimates using $c_1(\tau)$ and $\Delta\epsilon(\tau)$ for the \3{100} and \4{100} ensembles. 
The estimate using $\alpha_{c_1}$ is almost identical with the one obtained using \SPD in \cite{Stinner2015a}.
Note that error probability prediction using the $L\Delta\epsilon(\tau)$ provides slightly overconfident estimates, despite it still correctly captures the slope of the simulated error probability curve.
\begin{table}[t]
\centering
\caption{\BEC \BP threshold $\epsilon^*$, $\alpha_{c_1}$, $\alpha_{\Delta\epsilon}$, $\Theta$, $\tau_{\text{eff}}$ and $\tau_{\text{corr}}$ for \SCLDPC code ensembles with $L=100$.
}
\label{tab:alphas}
\begin{tabular}{l|ccccccc}
& $\epsilon^*$ & $\alpha_{c_1}$ & $\alpha_{\Delta\epsilon}$ & $\Theta$ & $\tau_{\text{eff}}$ & $\tau_\text{corr}$\\
\hline
$\3{100}$ & $0.4881$ & $6.049$ & $6.376$ & $2.0$ & $4.73$ & $1.37$ \\
$\4{100}$ & $0.4977$ & $5.112$ & $5.390$ & $2.0$ & $6.96$ & $0.91$ \\
\end{tabular}
\end{table}

%
\section{Summary and Conclusions}
We showed that \acf{PPD} and \acf{BP} resolve exactly the same \acf{VN} in every iteration.
Differences between code ensembles can also be observed with \PPD.
Instead of analyzing $c_1(\tau)$, $\Delta\epsilon(\tau)$ can also be used.
This is a significant complexity reduction since it can be derived from the less complex \BP analysis.
For further research, it will be interesting to obtain further moments of $\Delta\epsilon(\tau)$ analytically.
A more profound link between the bounds of Aref et. al for the convergence speed and $\Delta\epsilon$ in this work would also be interesting.
\section{Acknowledgement}
The authors would like to thank Gerhard Kramer and Michael Lentmaier, for the fruitful discussions.


%
\appendices
\section{Equivalence of \PPD and \BP}\label{sec:appendix:parpdvsbp}
We first show that without stopping sets, both \PPD and \BP recover exactly the same erased \VNs using messages only passed forwards.
We will then show that the message sent back using the \PPD does not resolve any additional \VNs, and thus the two decoders resolve the identical \VNs in every iteration if applied to the same residual graph.

\subsection{Forward Recovery of \VNs}
We do not consider $\msg{v}{c}(\ell+1)$ sent back from \VN $v$ to \CN $c$ as a function of $\msg{c}{v}(\ell)$ so that this analysis holds for both decoders.

Assume a residual graph without any stopping set \S.
Denote by $\Set{C}_i(\ell)$ the set of all \CNs in the residual graph connected to unknown \VNs $i$ times in iteration $\ell$.
We divide all \CNs in the residual graph into classes:
\begin{itemize}
 \item $\Set{C}_1(\ell)$: \CNs connected to unknown \VNs once in iteration $\ell$,
 \item $\Set{C}_{\geq2}(\ell)$: \CNs connected to unknown \VNs two or more times.
\end{itemize}

Since the \CNs in $\Set{C}_1(\ell)$ are connected to only one unknown \VN, both decoders can resolve these \VNs in iteration $\ell$.
\VNs connected only to $\Set{C}_{\geq2}(\ell)$ cannot be resolved by neither of the two decoders.
Therefore, these \VNs do not change their state and remain erased in iteration $\ell$.

\paragraph{First iteration}
Since the residual graph does not contain any stopping set, it must be resolvable with any of these decoders according to \cite{Di2002}.
Thus, there must be at least one \CN $c$ connected to only one erased \VN $v$, i.e. $\Set{C}_1(\ell)$ is non-empty.
An example is depicted in \figref{fig:messagessentintoresidualgraph}.

All other $\msg{v'}{c}(\ell),v'\in\Neighborsset{c}\setminus\{v\}$ apart from $\msg{v}{c}(\ell)$ are known from the transmission.
Since $\msg{c}{v}(\ell)=f(\msg{v'}{c}(\ell))$, $v$ cannot get erased again in any later iterations $\ell+i$.
If there is another \CN $c'\in\Neighborsset{v}\setminus\{c\}$ , $\msg{v}{c'}(\ell+1)$ will be resolved.
Since $\msg{v}{c'}(\ell)$ is resolved as soon as \emph{any} other \CN $\in\Neighborsset{v}$ sends a resolved message to $v$, it will stay resolved independently of the state of the other \CNs $\in\Neighborsset{v}$ in any later iterations $\ell+i$.
If $c'$ has degree $j$ in the first iteration and is connected to $k$ resolved \VNs, $c'\in\Set{C}_{j-k}(2)$ of the following iteration.

\begin{figure}[htb]
\footnotesize
\centering
\includegraphics{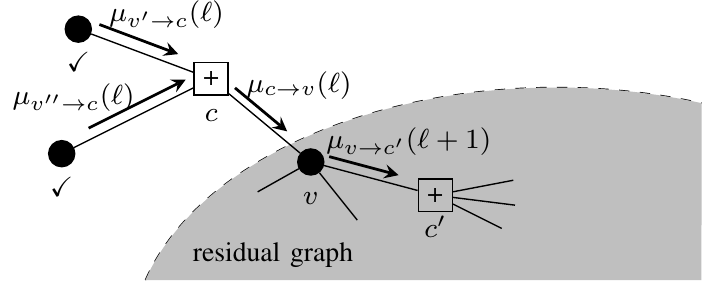}
\caption{\VN $v$ resolved by \CN $v$ forwarding the solution.}
\label{fig:messagessentintoresidualgraph}
\end{figure}

\paragraph{$\ell-th$ iteration}
If $\Set{C}_1(\ell)$ is non-empty and contains any \CN $c$, the connected unknown \VNs are resolved.
If there is any \CN $c'\in\Set{C}_{j}(\ell)$, $j>1$, connected to $k$ such resolved \VNs, $c'$ will be in $\Set{C}_{j-k}(\ell+\!1)$. 
$\msg{c'}{v}(\ell+\!1)$ will again stay resolved since they are a function of $c\in\Set{C}_1(\ell)$ whose messages cannot become erased again. 

If $\Set{C}_1(\ell)$ is empty, the decoding process does not resolve any further \VNs and stops.
The residual graph does not contain any \VNs any more since without any \S, it would contain at least a \CN connected to only one erased \VN. 

\paragraph{Graphs with stopping sets}
The set $\Set{C}_{\geq2}(\ell)$ is defined as the set that contains all \CNs connected to remaining unknown \VNs at least twice.
This set also includes \CNs connected to stopping sets.
Denote \CNs connected to \S with $\overline{\S}$.
Since messages sent from $\overline\S$ to the rest of the residual graph can never be resolved during the decoding process, we replace $\overline\S$ with erased messages sent from $\overline\S$ to the rest of the residual graph as illustrated with an example in \figref{fig:residualstoppingset}.

\begin{figure}[htb]
\footnotesize
\centering
\hspace*{1.28cm}
\includegraphics{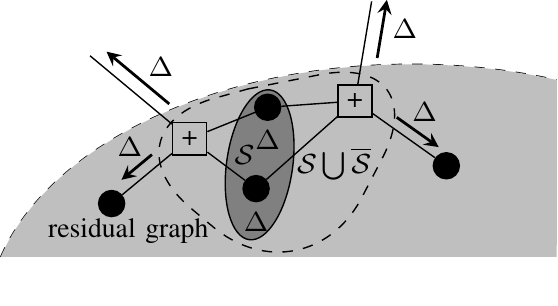}
\caption{Residual graph including $\overline\S$ and \S.}
\label{fig:residualstoppingset}
\end{figure}

The \CNs in $\overline\S$ are a part of $\Set{C}_{\geq2}(\ell)$ and therefore do not change and stay erased during all iterations.
Using any of the two decoders, the decoding of the \VNs of the residual graph proceeds identically until $\Set{C}_1(\ell)$ is empty.

\subsection{Backward messages of \VNs }
The decoding of any part of the residual graph depends only on messages sent forward into the residual graph of the respective iteration.
Since \VNs outside the residual graph are already known, messages sent back outside the residual graph cannot recover any additional \VNs.

\begin{theorem}[Exactly Equivalent Recovery]
\PPD and \BP recover exactly the same \VNs at each iteration.
\end{theorem}
\begin{proof}
Using \PPD, the message fed back also changes to resolved once a \VN is known.
The decoding proceeds by resolving the residual graph per iteration which does not depend on the message fed back to the graph outside the residual graph.
Thus, \PPD and \BP recover exactly the same \VNs in every iteration.
\end{proof}

\subsection{Observation for messages fed back}
Using \BP, messages sent back outside the graph are only resolved if a \VN is connected to two \CNs of $\Set{C}_1(\ell)$ in any iteration $\ell$.

\section{Expected Graph Evolution in a Single Iteration of \PPD}\label{sec:appendix:meanstepparallel}
Assume we know the graph \DD  $\{\l_{\v{v}}(\ell), r_{\v{c}}(\ell)\}_{\v{v}\in\mathcal{F}_v, \v{c}\in\overline{\mathcal{F}}_c}$ at a particular time $\ell$.
We compute the expected graph \DD evolution for the next iteration
\begin{align}
\E[R_{\v{c}}(\ell+1)-R_{\v{c}}(\ell)\Big|\{\l_{\v{v}}(\ell), r_{\v{c}}(\ell)\}_{\v{v}\in\F_v, \v{c}\in\overline{\F}_c}].
\end{align}

Using the \PPD, only \degone \CNs can be removed \emph{directly} and in an iteration, we remove all \degone \CNs along with all \VNs connected to them and edges attached to these \VNs. 
Removing these edges resolves also edges for some other \CNs and the \VNs are connected with as well, so that the type of such a \CN is modified from $\v{c}_1\in\overline{\mathcal{F}}_c$ to $\v{c}_2\in\overline{\mathcal{F}}_c$.
We therefore have an \emph{indirect} removal of a \CN of type $\v{c}_1$ from the residual graph and an insertion of a \CN of type $\v{c}_2$.
We assume that by lifting the constructions as described in Section \ref{sec:construction} we obtain independent edges.

The calculation of the graph evolution is outlined as follows:
  \begin{itemize}
  \item For each \VN type, calculate the number of edges of an edge type being resolved \emph{directly}.
  Take into account that all edges connected to this \VN type are in the residual graph and an edge can be resolved by multiple \degone \CNs.
  \item For each non-\degone \CN type in the residual graph, calculate the probability of being connected to any resolved \VN and resolving the connecting edges \emph{indirectly}.
  \end{itemize}

\paragraph{Directly resolved edge types}
Any \degone \CN of type $\v{e}_j,j\in[1,m]$ is directly removed from the graph as depicted in \figref{fig:parallelPDdirectresolving}.
To track the changes in the graph, we change to an edge perspective.
Denote with $\vpdir(\ell)=(\Pdir{1},\dots,\Pdir{m})$ the probability that an edge of a particular type is directly removed from the graph.
Assume that messages along different edges are independent due to the lifting with large $N$.
$\Pdir{j}$ can be calculated with 
\begin{align}\label{eq:direct}
\Pdir{j}=\frac{r_{\v{e}_j}(\ell)}{\displaystyle \sum_{\v{c}:c_j>0} r_{\v{c}}(\ell)},
\end{align}
for $j=1,\ldots,m$, which we define to be $0$ for $r_{\v{e}_j}(\ell)=0$.

\paragraph{Resolved \VN types}
A \VN of type \v{v} is resolved if it is connected to at least one \degone \CN.
Having $\l_\v{v}(\ell)$ \VNs of a specific type, we obtain the number of resolved \VNs of type \v{v} with 
\begin{align}
\l_\v{v}(\ell)\left(1-(\v{1}-\vpdir(\ell))^{\v{v}}\right).
\end{align}
Note that we assume that the probability that an edge type is directly removed is independent of the removal of other edges in the graph.

\paragraph{Indirectly resolved edge types}
Denote with \Pindir{j} the probability that an edge of type $j$ was resolved indirectly.
These edges of type $j$ are not connected to a \degone \CN but to any other possible \CN of type \v{c} in the residual graph with $c_j>0$:
\begin{align}\label{eq:Pindirj}
\Pindir{j}
&=(1-\Pdir{j})\cdot \left(1-(\v{1}-\vpdir(\ell))^{\v{v}-\v{e}_j}\right).
\end{align}
The total number of indirectly resolved edges of type $j$ in the residual graph is obtained with
\begin{align}
\sum_{\v{c}:c_j>0}\r_\v{c}\Pindir{j}.
\end{align}

\paragraph{Indirectly resolved edges of \CN types}
Denote with $\vpindir{\v{c}}(\ell)=(\Pindir{\v{c},1},\dots,\Pindir{\v{c},m})$ the probabilities that a \CN of type \v{c}, $|\v{c}|>1$, is connected to a resolved \VN via a specific edge type $j=1,\dots,m$.
Using \eqref{eq:direct} and \eqref{eq:Pindirj}, we calculate $\Pindir{\v{c},j}$ with
\begin{align}
 \Pindir{\v{c},j}&=\frac{\r_\v{c}}{\displaystyle \sum_{\substack{\v{c}':c_j>0\\|\v{c}'|>1}}\r_{\v{c}'}}
 \cdot\Pindir{j}\nonumber\\
  &=\frac{\r_\v{c}}{\displaystyle \sum_{\substack{\v{c}':c_j>0\\|\v{c}'|>1}}\r_{\v{c}'}}
 \cdot \frac{ \sum_{\v{c}:c_j>0} r_{\v{c}}(\ell)-r_{\v{e}_j}(\ell)}{\displaystyle \sum_{\v{c}:c_j>0} r_{\v{c}}(\ell)}\nonumber\\
 &\cdot\left(1-(\v{1}-\vpdir(\ell))^{\v{v}-\v{e}_j}\right)
 \nonumber\\
 %
 &=\frac{\r_\v{c}}{\displaystyle \sum_{\v{c}':c_j>0}\r_{\v{c}'}}
 \cdot \left(1-(\v{1}-\vpdir(\ell))^{\v{v}-\v{e}_j}\right).
\end{align}
An indirect resolving of an edge connected to a \CN of type \v{c} is illustrated in \figref{fig:parallelPDdirectresolving}.
We can now calculate the number of indirectly resolved \CNs of type \v{c} which we denote with \Eminus{\v{c}}:
\begin{align}\label{eq:ge1s:eindirc}
\Eminus{\v{c}}
&=r_{\v{c}}(\ell)(\vpindir{\v{c}}(\ell))^{\v{c}}.
\end{align}

\begin{figure}[htb]
\footnotesize
\centering
\includegraphics{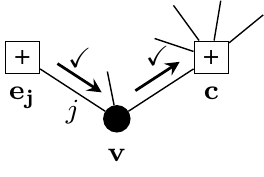}
\caption{\VN \v{v} connected to a \CN of type $e_j$ resolving an edge of type $j$, in turn resolving a socket of a \CN of type $c$.}
\label{fig:parallelPDdirectresolving}
\end{figure}

Since it is also possible that not all edges of a \CN are indirectly resolved, we obtain the fraction of \CNs of type $\v{c}$ reduced to a \CN of type $\v{c}'$ with
\begin{align}
 \Eminus{\v{c},\v{c}'}
 =&r_{\v{c}}(\ell)(\vpindir{\v{c}}(\ell))^{\v{c}-\v{c}'}(\v{1}-\vpindir{\v{c}}(\ell))^{\v{c}'}.\label{eq:ge1s:indirectlychangedtype}
\end{align}
This corresponds to resolving all other edges of \v{c} than
there are remaining in $\v{c}'$.

\pagebreak

\paragraph{Indirectly added \CN types}
\CNs of other types can be added to the residual graph if not all edges of an indirectly resolved \CN are resolved.
From \eqref{eq:ge1s:indirectlychangedtype}, we have
\begin{align}
\Eplus{\v{c}'}=\sum_{\v{c}\in\overline{\F}_c}\Eminus{\v{c},\v{c}'}.
\end{align}

\paragraph{Expected evolution in a single \PPD iteration}
We now calculate the expected evolution of the number of \CNs of type $\v{c}$ in the graph in a single \PPD iteration.
For any \CN type $\v{c}\in\overline{\F}_c$ such that $\v{c}\neq \v{e}_j,j\in[1,m]$, the overall evolution is 
\begin{align}
\E[R_{\v{c}}(\ell+1)-R_{\v{c}}(\ell)]=\Eplus{\v{c}}-\Eminus{\v{c}}.
\end{align}
For any \degone \CN of type $\v{e}_j,j\in[1,m]$, we have 
\begin{align}
\E[R_{\v{e}_j}(\ell+1)-R_{\v{e}_j}(\ell)]=\Eplus{\v{e}_j}-r_{\v{e}_j}(\ell).
\end{align}
\section{Density Evolution for \BP}\label{sec:appendix:DE}
We transmit over a \BEC.
We represent each symbol transmitted with a \VN which is erased with probability $\epsilon$.
Since \PPD and \BP are equivalent, $\Delta\epsilon$ can also be obtained from the analysis of \BP, for which we apply \DE \cite{Richardson2008}. 
Denote the erasure probability of the channel with $\epsilon_{ch}$.
The erasure probability is tracked for each edge type of the protograph separately.
For any edge type $s\in\{1,\dots,m\}$, we denote the erasure probability of messages sent from the \VN in iteration $\ell$ with $x_s(\ell)$ and from the \CN with $y_s(\ell)$.
We denote them as vectors with
\begin{align}
\v{x}(\ell)&=(x_1(\ell),\dots,x_m(\ell)),\\
\v{y}(\ell)&=(y_1(\ell),\dots,y_m(\ell)).
\end{align}
Since we calculate edge type erasure probabilities, we do not consider the additional \CN types after transmission in the residual graph and we consider the initial \CN types $\v{c}\in\F_c$.
Using the constructions described in Section \ref{sec:construction}, the edge type $s$ is only connected to one \VN type $\v{v}\in\F_v$ and one \CN type $\v{c}\in\F_c$.
The \CN and \VN functions can be calculated exactly with
\begin{align}
 y_s(\ell)&=1-(\v{1}-\v{x}(\ell-1))^{\v{c}-\v{0}_{\sim s}},&\v{c}\in\F_c,\\
 x_s(\ell)&=\epsilon\v{y}(\ell)^{\v{v}-\v{0}_{\sim s}},&\v{v}\in\F_v,
\end{align}
where all exponents are $\geq 0$ and punctured \VNs types are initialized with $\epsilon=0$.
Denote the erasure probability of a \VN type $\v{v}\in\F_v$ after the $\ell$-th iteration 
\begin{align}
\epsilon_{\v{v}}(\ell)&=\epsilon_{ch}\v{y}(\ell)^{\v{v}},&\v{v}\in\F_v.
\end{align}
We calculate the erased \VNs per iteration with
\begin{align}
 \epsilon(\ell)=\sum_{\v{v}\in\F_v}\l_\v{v}\epsilon_\v{v}(\ell).
\end{align}
As shown in \cite{Richardson2008}, $\epsilon_{\v{v}}$ for every \VN type $\v{v}\in\F_v$ is monotonically decreasing during the decoding process.

\end{document}